\documentclass[11pt, reqno]{amsart}
\usepackage{microtype}[draft]
\usepackage[scr=euler]{mathalfa}
\usepackage{stmaryrd}
\usepackage{tikz-cd}
\usepackage{amsmath, amssymb}
\usepackage{mathtools}

\usepackage[hyphens]{url}
\usepackage[breaklinks=true]{hyperref}

\usepackage[framemethod=TikZ]{mdframed}
\mdfsetup{%
middlelinecolor=red,
middlelinewidth=2pt,
backgroundcolor=red!10,
roundcorner=10pt}

\usepackage{comonads}

\title{Notes on presheaf representations of strategies \\
and cohomological refinements of $k$-consistency and $k$-equivalence}

\author{Samson Abramsky}
\address{Department of Computer Science, University College London, 66-72 Gower St., London WC1E 6EA, U.K.}\email{s.abramsky@ucl.ac.uk}

\begin{document}

\begin{abstract}
In this note, we show how positional strategies for  $k$-pebble games have a natural representation as certain presheaves.
These representations  correspond exactly to the sheaf-theoretic models of contextuality introduced by Abramsky-Brandenburger.
We study the notion of cohomological $k$-consistency recently introduced by Adam \'O Conghaile from this perspective.
\end{abstract}

\maketitle

\section{Introduction}
These are working notes, originally written in September-November 2021. Adam \'O Conghaile has proposed in \cite{AOC2021} some very interesting ideas on how to leverage ideas from the sheaf-theoretic approach to contextuality \cite{abramsky2011sheaf,DBLP:conf/csl/AbramskyBKLM15} to obtain cohomological refinements of $k$-consistency in constraint satisfaction, and of the Weisfeiler-Leman approximations to structure isomorphism. The purpose of these notes was to develop the sheaf-theoretic aspects of these ideas more fully, and to give a more uniform and conceptual account of the cohomological refinements.

They are intended to feed into  ongoing joint work with Adam, Anuj Dawar, and Rui Soares Barbosa which aims to get a deeper understanding and more extensive results on this approach.

\subsection*{Notation}

We will not distinguish notationally  in this note between a $\sg$-structure $A$ and its underlying universe or carrier, also written as $A$.

\textbf{Stylistic note} When we say ``there is a bijective correspondence between things of type $X$ and things of type $Y$'', we mean that there are explicit, inverse transformations between descriptions of things of these types, and hence the two notions are essentially the same.

We fix a finite relational vocabulary $\sg$, and a finite $\sg$-structure $B$. 
We use the notation $A \to B$ to mean that there exists a homomorphism from $A$ to $B$, and $A \nto B$ for the negation of this statement.
The \emph{constraint satisfaction problem} $\CSP(B)$ is to decide, for an instance given by a finite $\sg$-structure $A$, whether there is a homomorphism $A \to B$.
We refer to $B$ as the \emph{template}.

\section{Positional strategies and presheaves}

Given a finite $\sg$-structure $A$, a \emph{positional strategy} for the existential $k$-pebble game from $A$ to $B$ is given by a family $S$ of homomorphisms $f : C \to B$, where $C$ is an induced substructure of $A$ with $|C| \leq k$. This is subject to the following conditions:
\begin{itemize}
\item \textbf{down-closure}: If $f : C \to B \in S$ and $C' \subseteq C$, then $f |_{C'} : C' \to B \in S$.
\item \textbf{forth condition}: If $f : C \to B \in S$, $|C| <k$, and $a \in A$, then for some $f' : C \cup \{a\} \to B \in S$, $f' |_{C} = f$.
\end{itemize}
Duplicator wins the existential $k$-pebble game iff there is a non-empty such  strategy.
It is well-known that the existence of such a strategy is equivalent to the \emph{strong $k$-consistency} of $A$ in the sense of constraint satisfaction \cite{kolaitis2000game}. Moreover, strong $k$-consistency can be determined by a polynomial-time algorithm.

We write $\Sk(A)$ for the poset of subsets of $A$ of cardinality $\leq k$. Each such subset gives rise to an induced substructure of $A$.
We define a presheaf $\Hk : \Sk(A)^{\op} \to \Set$ by $\Hk(C) = \hom(C,B)$. If $C' \subseteq C$, then the restriction maps are defined by $\rho^{C}_{C'}(h) = h |_{C'}$.
This is the \emph{presheaf of partial homomorphisms}.\footnote{We can make this into a sheaf, analogous to the event sheaf in \cite{DBLP:conf/csl/AbramskyBKLM15}.}

A subpresheaf of $\Hk$ is a presheaf $\Sh$ such that $\Sh(C) \subseteq \Hk(C)$ for all $C \in \Sk(A)$, and moreover if $C' \subseteq C$ and $h \in \Sh(C)$, then $\rho^{C}_{C'}(h) \in \Sh(C')$.
A presheaf is \emph{flasque} (or ``flabby'') if the restriction maps are surjective. This means that if $C \subseteq C'$, each $h \in \Sh(C)$ has an extension $h' \in \Sh(C')$ with $h' |_{C} = h$.
\begin{proposition}
There is a bijective correspondence between
\begin{enumerate}
\item positional strategies from $A$ to $B$
\item flasque sub-presheaves of $\Hk$.
\end{enumerate}
\end{proposition}
\begin{proof}
The property of being a subpresheaf of $\Hk$ is equivalent to the down-closure property, while being flasque is equivalent to the forth condition.
\end{proof}

Note that these flasque subpresheaves correspond to the ``empirical models'' studied in \cite{DBLP:conf/csl/AbramskyBKLM15} in relation to contextuality.

We can define a union operation on sub-presheaves of $\Hk$ pointwise: $(\bigcup_{i \in I} \Sh_i)(C) = \bigcup_{i \in I} \Sh_i (C)$. This preserves the property of being a flasque subpresheaf of $\Hk$. Thus there is a largest such sub-presheaf, which we denote by $\Sbar$. The following is then immediate:
\begin{proposition}
The following are equivalent:
\begin{enumerate}
\item $A$ is strongly $k$-consistent.
\item $\Sbar \neq \es$.
\end{enumerate}
\end{proposition}

\section{Presheaf representations and  the pebbling comonad}\label{pebbsec}

\begin{proposition}\label{pebbprop}
The following are equivalent:
\begin{enumerate}
\item There is a coKleisli $I$-morphism $\Pk A \to B$
\item There is a non-empty flasque sub-presheaf $\Sh$ of $\Hk$.
\item $\Sbar \neq \es$.
\end{enumerate}
\end{proposition}
\begin{proof}
Given a coKleisli morphism $f : \Pk A \to B$, each $s \in \Pk A$ determines a subset $C \in \Sk(A)$ of those elements with ``live'' pebbles placed on them, and the responses of $f$ to the prefixes of $s$ determines a homomorphism $h : C \to B$, as described in detail in \cite{abramsky2017pebbling}. We define $\Sh_f(C)$ to be the set of all such morphisms, extending this with $\Sh_f(\es) = \es$.
The forth property follows from the fact that if $|C| < k$, then we have a free pebble we can place on an additional element, and the response of $f$ to this extension of the sequence gives the required extension of $h$.

Somewhat surprisingly, the down-closure property is slightly more subtle. If $C$ is a singleton, this follows trivially from our stipulation on the empty set. Otherwise, we can remove an element $a$ from $C$ by moving the pebble currently placed on it to another element $a'$ of $C$. This means we have duplicate pebbles placed on $a'$. Because $f$ is an $I$-morphism, it must give the same response to this duplicate placing as it did when the previous pebble was placed on $a'$, and we obtain the restriction of $h$ as required.

Now suppose that we have $\Sh$. For each $s \in \Pk A \cup \{ []\}$, we define a homomorphism $h_s : C \to B$, where $C$ is the set of elements with live pebbles in $s$, by induction on $|s|$.
When $s = []$, $h_s = \es$.
Given $s[(p,a)] \in \Pk A$, by induction we have defined $h_s : C \to B$. There are several cases to consider, depending on the set of live elements $C'$ corresponding to $s[(p,a)]$:
\begin{enumerate}
\item If $C' = C \cup \{a\}$, then since $\Sh$ is flasque, there is $h'  \in \Sh(C')$ with $h' |_{C} = h_s$. We define $h_{s[(p,a]]} = h'$.
\item It may be the case that $C' = C \setminus \{b\}$ for some $b$. This can happen if $a \in C$, so $p$ is a duplicate placing of a pebble on $a$, and the pebble $p$ was previously placed on $b$. 
In this case, we define $h_{s[(p,a]]} = h_s |_{C'}$.
\item Finally, we may have $C' = (C \setminus \{b\}) \cup \{a\}$. In this case, $h_s |_{C \setminus \{b\}} \in \Sh(C \setminus \{b\})$, and since $\Sh$ is flasque, has an extension $h'  \in \Sh(C')$.
We define $h_{s[(p,a]]} = h'$.
\end{enumerate}
Now we define $f_{\Sh}(s[(p,a)]) = h_{s[(p,a)]}(a)$.

The equivalence of (2) and (3) is immediate.
\end{proof}
Since choices were made in defining $f_{\Sh}$, the passages between $f$ and $\Sh_{f}$ and $\Sh$ and $f_{\Sh}$ are not mutually inverse. Using the more involved argument in \cite[Proposition 9]{abramsky2017pebbling}, we can get a closer correspondence.

\section{Local consistency as coflasquification}

Seen from the sheaf-theoretic perspective, the local consistency algorithm has a strikingly simple and direct mathematical specification.

Given a category $\CC$, we write $\Pshv{\CC}$ for the category of presheaves on $\CC$, \ie functors $\CC^{\op} \to \Set$, with natural transformations as morphisms. We write $\Pshvf{\CC}$ for the full subcategory of flasque presheaves.

\begin{proposition}
The inclusion $\Pshvf{\CC} \inc \Pshv{\CC}$ has a right adjoint, so the flasque presheaves form a coreflective subcategory.
The associated idempotent comonad on $\widehat{\Sk(A)}$ is written as $\Sh \mapsto \Shfl$, where $\Shfl$ is the largest flasque subpresheaf of $\Sh$. The counit is the inclusion $\Shfl \inc \Sh$, and idempotence holds since $\Sh^{\Diamond\Diamond} = \Shfl$. We have  $\Hk^{\Diamond} = \Sbar$.
\end{proposition}
\begin{proof}
The argument for the existence of a largest flasque subpresheaf of a given presheaf is similar to that given in the previous section: the empty presheaf is flasque, and flasque subpresheaves are closed under unions.\footnote{Note that restriction is well-defined  on unions of subpresheaves of a given presheaf $\Sh$, which represent joins in the subobject lattice $\mathbf{Sub}(\Sh)$.} The key point for showing couniversality is that the image of a flasque presheaf under a natural transformation is flasque. Thus any natural transformation $\Sh' \natarrow \Sh$ from a flasque presheaf $\Sh'$  factors through the counit inclusion $\Shfl \inc \Sh$.
\end{proof}
This construction amounts to forming a greatest fixpoint. In our concrete setting, the standard local consistency algorithm (see e.g.~\cite{barto2014constraint}) builds this greatest fixpoint by filtering out elements which violate the restriction or extension conditions.

We shall give an explicit description of this construction in terms of presheaves, as it sets a pattern we shall return to repeatedly. Give a family $\{\Sh(C) \subseteq \Hk(C)\}_{C \in \Sk(A)}$\footnote{Not necessarily a presheaf, since it need not be closed under restriction.}, we define $\Shup$ by 
\[ \Shup(C) := \begin{cases}
\{ s \in \Sh(C) \mid  \forall a \in A. \, \exists s' \in \Sh(C \cup \{a\}). \, s' |_C = s \} & |C| < k \\
\Sh(C) & \text{otherwise}
\end{cases}
\]
and $\Shdown$ by
\[ \Shdown(C) := 
\{ s \in \Sh(C) \mid  \forall a \in C. \,  s |_{C \setminus \{a\}} \in \Sh(C \setminus \{a\}) \} 
\]
Clearly these constructions are polynomial in $|A|$, $|B|$. We can then define an iterative process
\[ \Hk \linc \Hk^{\uparrow\downarrow} \linc \cdots \linc \Hk^{(\uparrow\downarrow)^{m}} \linc \cdots \]
Since the size of $\Hk$ is polynomially bounded in $|A|, |B|$, this will converge to a fixpoint in polynomially many steps.
This fixpoint is $\Hk^{\Diamond} = \Sbar$.

This construction is dual to a standard construction in sheaf theory,  which constructs a flasque sheaf extending a given sheaf, leading to a monad, the \emph{Godement construction} \cite{godement1958topologie}.

The following proposition shows how this comonad propagates \emph{local inconsistency} to \emph{global inconsistency}.
\begin{proposition}
Let $\Sh$ be a presheaf on $\Sk(A)$.
If $\Sh(C) = \es$ for any $C \in \Sk(A) \setminus \{ \es \}$, then $\Shfl = \es$.
\end{proposition}
\begin{proof}
If is sufficient to show that $\Shfl(\{a\}) = \es$ for all $a \in A$, since the only presheaf for which this can hold is the empty presheaf. This will be the case for all $a \in C$ by propagating the forth condition.\footnote{A more precise argument goes by induction on $|C|$.} Now fix $a \in C$. For any $b \in A \setminus C$,
consider $\Shfl(\{ a,b \})$. This must be empty since it otherwise it would violate the restriction condition to $\{ a \}$.
But then $\Shfl(\{ b \})$ must be empty, since otherwise it would violate the forth condition for $a$.
\end{proof}

\section{Global sections and compatible families}
\label{secglobal}

A global section of a flasque subpresheaf $\Sh$ of $\Hk$ is a natural transformation $\One \natarrow \Sh$. More explicitly, it is a family $\{ h_{C} \}_{C \in \Sk(A)}$ with $h_C \in \Sh(C)$ such that, whenever $C \subseteq C'$, $h_C = h_{C'} |_{C}$. 
\begin{proposition}\label{homglobsecsprop}
Suppose that $k \geq n$, where $n$ is  the maximum arity of any relation in $\sg$. There is a bijective correspondence between
\begin{enumerate}
\item homomorphisms $A \to B$
\item global sections of $\Sbar$.
\end{enumerate}
\end{proposition}
\begin{proof}
Given a homomorphism $h : A \to B$, the family $\{ h |_{C} \}_{C \in \Sk(A)}$ is a flasque subpresheaf of $\Hk$, and hence is a subpresheaf of $\Sbar$. There is a corresponding global section of $\Sbar$, which picks out $h |_{C}$ at each $C \in \Sk$.

Conversely, given a global section $\{ h_{C} \}_{C \in \Sk(A)}$, we can define $h : A \to B$ by $h(a) = h_{\{a\}}(a)$. We must show that $h$ is a homomorphism. Given a relation instance $R^A(a_1,\ldots ,a_n)$, since $k \geq n$, $C := \{ a_1, \ldots , a_n \} \in \Sk(A)$, and for each $i$, $h_C |_{\{a_i\}}  = h_{\{a_i\}}$, so $h_C(a_i) = h(a_i)$. Since $h_C$ is a homomorphism, the relation instance is preserved by $h$.
\end{proof}

\textbf{Fixed assumption}
Henceforth, we shall always make the background assumption that $k \geq n$, where $n$ is the relational width of $\sg$.

Another representation of global sections will be very useful. This will focus on the \emph{maximal elements} $\Mk(A)$ of the poset $\Sk(A)$, \ie~the $k$-element subsets. 
A \emph{$k$-compatible family} in $\Sbar$ is a family $\{ h_C \}_{C \in \Mk(A)}$ such that, for all $C, C' \in \Mk(A)$, 
\[ \rho^{C}_{C \cap C'}(h_C) = \rho^{C'}_{C \cap C'}(h_C') . \]

\begin{proposition}
\label{compfamprop}
There is a bijective correspondence between global sections  and $k$-compatible families of $\Sbar$.
\end{proposition}
\begin{proof}
Clearly the restriction of a global section to $\Skm(A)$ gives a $k$-compatible family, since $ h_{C} |_{C \cap C'} = h_{C \cap C'} = h_{C'} |_{C \cap C'}$.

Conversely, given a $k$-compatible family $\{ h_{C} \}_{C \in \Skm(A)}$ and any $C' \in \Sk(A)$, we define $h_{C'} = h_{C_1} |_{C'}$, where $C' \subseteq C_1 \in \Skm(A)$. If $C' \subseteq C_2 \in \Skm(A)$, then $C' \subseteq C_1 \cap C_2$, and 
\[ h_{C_1} |_{C'} = (h_{C_1} |_{C_1 \cap C_2})|_{C'} = (h_{C_2} |_{C_1 \cap C_2})|_{C'} = h_{C_2} |_{C'}, \]
so this definition is consistent across maximal extensions of $C'$, and yields a well-defined global section.
\end{proof}

\section{Cohomological $k$-consistency}\label{cohomkconsec}

We can summarize the results of Section~\ref{secglobal} as follows:
\begin{proposition} 
There is a polynomial-time reduction from $\CSP(B)$ to the problem, given any instance $A$, of determining whether the associated presheaf $\Sbar$ has a global section, or equivalently, a $k$-compatible family.
\end{proposition}
Of course, since $\CSP(B)$ is NP-complete in general, so is the problem of determining the existence of a global section.
This motivates finding an efficiently computable approximation.

This leads us to the notion of cohomological $k$-consistency recently introduced by Adam \'O Conghaile \cite{AOC2021}.
This leverages results from \cite{DBLP:journals/corr/abs-1111-3620,DBLP:conf/csl/AbramskyBKLM15} on using cohomology to detect contextuality, and applies them in the setting of constraint satisfaction.

\section{Background on contextuality}

The logical structure of quantum mechanics is given by a family of overlapping perspectives or contexts. Each context appears classical, and different contexts agree locally on their overlap. However, there is no way to piece all these local perspectives together into an integrated whole, as shown in many experiments, and proved rigorously using the mathematical formalism of quantum mechanics. 

To illustrate this non-integrated feature  of quantum mechanics, we may consider the well-known ``impossible'' drawings by Escher, such as the one shown in Figure~\ref{EE0}.  

\begin{figure}
\begin{center}
\includegraphics[width=0.64 \textwidth]{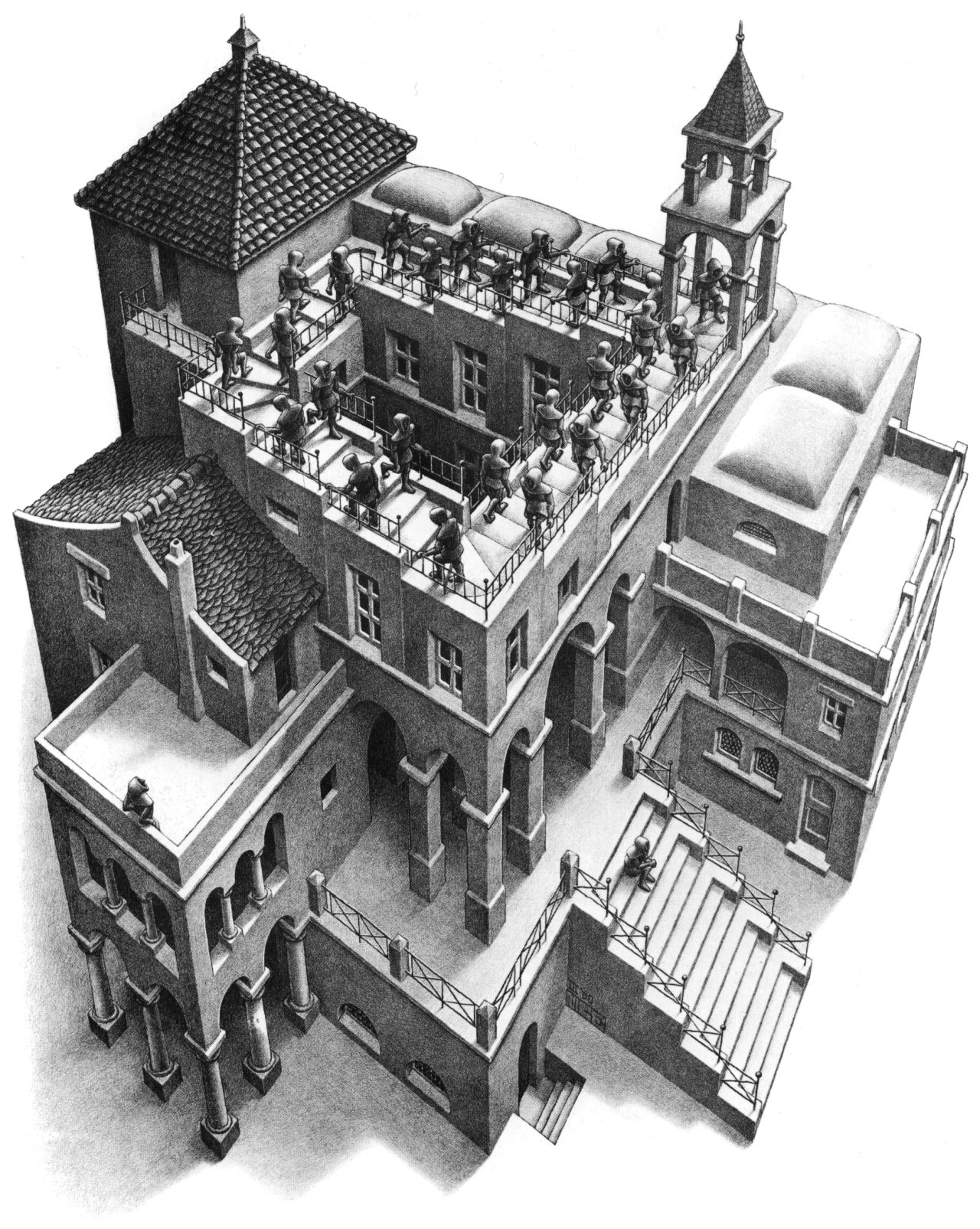}\\
\caption{M. C. Escher, \emph{Klimmen en dalen (Ascending and descending)}, 1960. Lithograph.}
\label{EE0}
\end{center}
\end{figure}

Clearly, the staircase \emph{as a whole} in Figure~\ref{EE0} cannot exist in the real world.  Nonetheless, the constituent parts of Figure~\ref{EE0} make sense \emph{locally}, as is clear from Figure~\ref{E1}.
Quantum contextuality shows that  the logical structure of quantum mechanics exhibits exactly these features of \emph{local consistency}, but \emph{global inconsistency}.
We note that Escher's work was inspired by the \emph{Penrose stairs} from \cite{pen15}.\footnote{Indeed, these figures provide more than a mere analogy. Penrose has studied the topological ``twisting'' in these figures using cohomology \cite{penrose1992cohomology}. This is quite analogous to our use of sheaf cohomology to capture the logical twisting in contextuality.}

\begin{figure}
\begin{center}
\includegraphics[width=0.7 \textwidth]{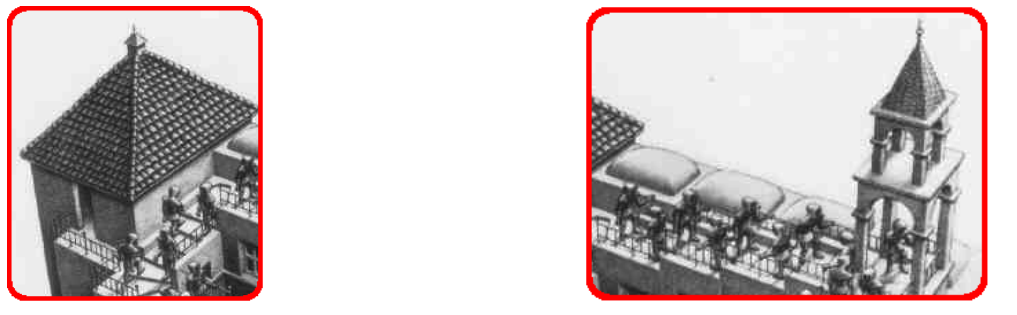}
\includegraphics[width=0.7 \textwidth]{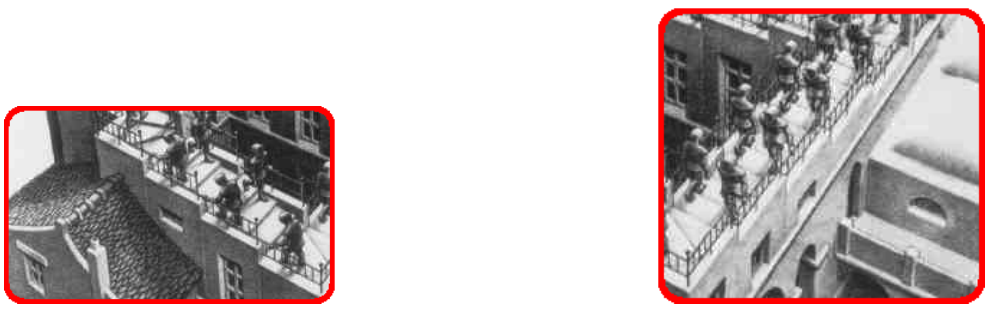}
\caption{Locally consistent parts of Figure~\ref{EE0}.}
\label{E1}
\end{center}
\end{figure}

\section{Brief review of contextuality, cohomology and paradox}

We begin by reviewing some background. In this note, we only consider commutative rings with unit, which we refer to simply as ``rings''.
Given a ring $R$, the category of $R$-modules is denoted $\RMod$. There is an evident forgetful functor $U : \RMod \to \Set$, and an adjunction 
\[ \begin{tikzcd}
\Set \arrow[r, bend left=25, ""{name=U, below}, "\FR"{above}]
\arrow[r, leftarrow, bend right=25, ""{name=D}, "U"{below}]
& \RMod
\arrow[phantom, "\textnormal{{$\bot$}}", from=U, to=D] 
\end{tikzcd}
\]
Here $\FR(X) = R^{(X)}$ is the free $R$-module generated by $X$, given by the formal $R$-linear combinations of elements of $X$, or equivalently the functions $X \to R$ of finite support, with pointwise addition and scalar multiplication. Given $f : X \to Y$, the functorial action of $\FR$ is given by $\FR f : \sum_i r_i x_i \mapsto \sum_j s_j y_j$, where $s_j = \sum_{f(x_i) = y_j} r_i$.
The unit of this adjunction $\eta_X : X \to R^{(X)}$ embeds $X$ in $R^{(X)}$ by sending $x$ to $1 \cdot x$, the linear combination with coefficient $1$ for $x$, and $0$ for all other elements of $X$.

Given $A$ with associated presheaf $\Sbar$, we can define the $\ZMod$-valued presheaf $\FZ \Sbar$.\footnote{Note that $\ZMod$ is isomorphic to $\AbGrp$, the category of abelian groups.} In  \cite{DBLP:journals/corr/abs-1111-3620,DBLP:conf/csl/AbramskyBKLM15} a cohomological invariant $\gamma$ is defined for a class of presheaves including $\FZ \Sbar$. 

\subsection*{The cohomological invariant from \cite{DBLP:journals/corr/abs-1111-3620,DBLP:conf/csl/AbramskyBKLM15}}

A cohomological invariant $\gamma$ is defined for a class of presheaves including $\FZ \Sbar$. 
\begin{itemize}
\item Given a flasque subpresheaf $\Sh$ of $\Hk$, we have the  $\AbGrp$-valued presheaf $\FF = \FZ \Sh$. We use the \Cech cohomology with respect to the cover $\MM = \Mk(A)$.
\item In order to focus attention at the context $C \in \MM$, we use the presheaf $\FF |_{C}$, which ``projects'' onto $C$. The cohomology of this presheaf is the \emph{relative cohomology} of $\FF$ at $C$.
The $i$'th relative \Cech cohomology group of $\FF$ is written as $\Cohom{i}$.
\item We have the \emph{connecting homomorphism} $\Cohom{0} \to \Cohom{1}$ constructed using the Snake Lemma of homological algebra.
\item The cohomological obstruction $\gamma : \FF(C) \to \Cohom{1}$ defined in \cite{DBLP:conf/csl/AbramskyBKLM15} is this connecting homomorphism, composed with the isomorphism $\FF(C) \cong \Cohom{0}$.
\end{itemize}

For our present purposes, the relevant property of this invariant is the following \cite[Proposition 4.2]{DBLP:journals/corr/abs-1111-3620}:
\begin{proposition}\label{gammacompfamprop}
For a local section $s \in \Sbar(C_0)$, with $C_0 \in \Mk(A)$, the following are equivalent:
\begin{enumerate}
\item $\gamma(s) = 0$
\item There is a $\ZZ$-compatible family $\{ \alpha_{C} \}_{C \in \Mk(A)}$ with $\alpha_C \in \FZ \Sbar(C)$, such that, for all $C, C' \in \Mk(A)$:
$\rho^{C}_{C \cap C'}(\alpha_C) =  \rho^{C'}_{C \cap C'}(\alpha_{C'})$. Moreover, $\alpha_{C_0} = 1\cdot s$.
\end{enumerate}
\end{proposition}
We call  a family as in (2) a \emph{$\ZZ$-compatible extension} of $s$. We can regard such an extension as a ``$\ZZ$-linear approximation'' to a homomorphism $h : A \to B$ extending $s$.

Note that, under the embedding $\eta$, every global section of $\Sbar$, given by a compatible family $\{ h_C \}_{C \in \Mk(A)}$, will give rise to a $\ZZ$-compatible family $\{ \eta_{C}(h_C) \}_{C \in \Mk(A)}$. 

\section{Algorithmic properties of the cohomological obstruction}

The crucial advantage of the cohomological notion in the present setting is given by the following observation from \cite{AOC2021}:
\begin{proposition}
There is a polynomial-time algorithm for deciding, given $s \in \Sbar(C)$, $C \in \Mk(A)$, whether $s$ has a $\ZZ$-compatible extension.
\end{proposition}
\begin{proof}
Firstly, the number of ``contexts'' $C \in \Mk(A)$ is given by $K := |\Mk(C)| = \binom{|A|}{k} \leq |A|^k$.
For each context $C$, the number of elements of $\Sbar(C)$ is bounded by $M := |B|^k$. Each constraint $\rho^{C}_{C \cap C'}(\alpha_C) =  \rho^{C'}_{C \cap C'}(\alpha_{C'})$ can be written as a set of homogeneous linear equations: for each $s \in \Sbar(C \cap C')$, we have the equation
\[ \sum_{\substack{t \in \Sbar(C),\\ t |_{C \cap C'} = s}} r_{C,t}  \;\;\; - \;\; \sum_{\substack{t' \in \Sbar(C'), \\ t' |_{C \cap C'} = s}} r_{C',t'} \; = \; 0 \]
in the variables $r_{C,s}$ as $C$ ranges over contexts, and $s$ over $\Sbar(C)$. Altogether, there are $\leq K^2M$ such equations, over $KM$ variables.
The constraint that $\alpha_{C_0} = 1\cdot s$ can be written as  a further $M$ equations forcing $r_{C_{0},s} = 1$, and $r_{C_{0},s'} = 0$ for $s' \in \Sbar(C) \setminus \{s\}$.
The whole system can be described by the equation $\mathbf{A}\mathbf{x} = \mathbf{v}$, where $\mathbf{A}$ is a matrix with entries in $\{ -1, 0, 1\}$, of dimensions $(K^2 +1)M \times KM$, which is of size polynomial in $|A|$, $|B|$, while $\mathbf{v}$ is a vector with one entry set to $1$, and all other entries to $0$. The existence of a $\ZZ$-compatible extension of $s$ is equivalent to the existence of a solution for this system of equations. Since solving systems of linear equations over $\ZZ$ is in PTIME \cite{kannan1979polynomial}, this yields the result.
\end{proof}

Given a flasque subpresheaf $\Sh$ of $\Hk$, and $s \in \Sh(C)$, $C \in \Mk(A)$, we write $\Ztest(\Sh,s)$ for the predicate which holds iff $s$ has a $\ZZ$-compatible extension in $\Sh$. 
The key idea in \cite{AOC2021} is to use this predicate as a filter to refine local consistency. 
We define $\Shch \inc \Sh$ by 
\[ \Shch(C) := \begin{cases}
 \{ s \in \Sh(C) \mid \Ztest(\Sh,s)  \}  & \text{$C \in \Mk(A)$} \\
 \Sh(C) & \text{otherwise}
 \end{cases}
 \]
Note that $\Shch$ can be computed with polynomially many calls of $\Ztest$, and thus is itself computable in polynomial time.
$\Shch$ is closed under restriction, hence a presheaf. It is not necessarily flasque.
Thus we are led to the following iterative process, which is the ``cohomological $k$-consistency'' algorithm from \cite{AOC2021}:
\[ \Hk \linc \Hk^{\fl} \linc \Hk^{\fl\ch\fl} \linc \cdots \linc \Hk^{\fl(\ch\fl)^{m}} \linc \cdots \]
Since the size of $\Hk$ is polynomially bounded in $|A|, |B|$, this will converge to a fixpoint in polynomially many steps.
We write $\Shm$ for the $m$'th iteration of this process, and $\Shst$ for the fixpoint. Note that $\Sbar = \Sh_{k}^{(0)}$.

Note that this process involves computing $\Ztest(\Sh, s)$ for the \emph{same} local section $s$ with respect to \emph{different} presheaves $\Sh$. 
\begin{proposition}
If $\Sh' \inc \Sh$ are flasque subpresheaves of $\Hk$, with $s \in \Sh'(C) \subseteq \Sh(C)$, $C \in \Mk(A)$, then 
$\Ztest(\Sh',s)$  implies $\Ztest(\Sh,s)$. The converse does not hold in general.
\end{proposition}
\begin{proof}
Deciding $\Ztest(\Sh,s)$ amounts to determining whether a set $E$ of $\ZZ$-linear equations has a solution. The corresponding procedure for $\Ztest(\Sh',s)$ is equivalent to solving a set $E \cup F$, where $F$ is a set of equations forcing some of the variables in $E$, corresponding to the coefficients of local sections in $\Sh$ which are not in $\Sh'$, to be $0$. A solution for $E \cup F$ will also be a solution for $E$, but solutions for $E$ may not extend to $E \cup F$.
\end{proof}

Returning to the CSP decision problem, we define some relations on structures:
\begin{itemize}
\item We define $A \tok B$ iff $A$ is\emph{ strongly $k$-consistent} with respect to $B$, \ie iff $\Sbar = \Sh_{k}^{(0)} \neq \es$.
\item We define $A \toZk B$ if $\Shst \neq \es$, and say that $A$ is \emph{cohomologically $k$-consistent} with respect to $B$.
\item We define $A \toZko B$ if $\Shone \neq \es$, and say that $A$ is \emph{one-step cohomologically $k$-consistent} with respect to $B$.
\end{itemize}

As already observed, we have:
\begin{proposition}
\label{ckconprop}
There are algorithms to decide $A \toZk B$ and $A \toZko B$ in time polynomial in $|A|$, $|B|$.
\end{proposition}

We can regard these relations as approximations to the ``true'' homomorphism relation $A \to B$. The soundness of these approximations is stated as follows:
\begin{proposition} 
\label{chainimpprop}
We have the following chain of implications:
\[ A \to B \IMP A \toZk B \IMP A \toZko B \IMP A \tok B .\]
\end{proposition}
\begin{proof}
For the first implication, each homomorphism $h : A \to B$ gives rise to a compatible family, and hence, as already remarked, to a $\ZZ$-linear family $\{ \eta_{C}(h_C) \}_{C \in \Mk(A)}$ which extends each of its elements. The second and third implications are  immediate from the definitions.
\end{proof}

\section{Composition}

So far we have focussed on a single template structure $B$ and a single instance $A$. We now look at the global structure.
We can define a poset-enriched category (or locally posetal 2-category) $\Ck$ as follows:
\begin{itemize}
\item Objects are relational structures.
\item 1-cells $\Sh : A \to B$  are  flasque subpresheaves of $\HkAB$.
\item The local posetal structure (2-cells) is given by subpresheaf inclusions $\Sh' \inc \Sh$. 
\end{itemize}
The key question is how to compose 1-cells. We define a 2-functor $\Ck(A,B) \times \Ck(B,C) \to \Ck(A,C)$. Given $\Sh : A \to B$ and $\Tsh : B \to C$, we define $\Tsh \circ \Sh : A \to C$ as follows: for $U \in \Sk(A)$, $\Tsh \circ \Sh(U) := \{ t \circ s \mid s \in \Sh(U) \AND t \in \Tsh(\im s) \}$. This is easily seen to be monotone with respect to presheaf inclusions.

This composition immediately allows us to conclude that the  $\tok$ relation is transitive. We now wish to extend this to show that $\toZk$ 
is transitive. For guidance, we can look at a standard construction, of \emph{group (or monoid) ring} \cite{lang2012algebra}. Given a monoid $M$, with multiplication $m : M \times M \to M$, and a commutative ring $R$, we form the free $R$-module $R^{(M)}$. The key point is how to extend the multiplication on $M$ to $R^{(M)}$. For this, we recall the tensor product of $R$-modules, and the fact that $\FR$ takes products to tensor products, so that $R^{(M)} \otimes R^{(M)} \cong R^{(M \times M)}$. We can use the universal property of the tensor product \cite{mac2013categories}
\[ \begin{tikzcd}
R^{(M)} \times R^{(M)} \ar[r] \ar[d, "\otimes"'] & R^{(M)} \ar[d, leftarrow, "R^{(m)}"] \\
R^{(M)} \otimes R^{(M)} \ar[r, "\cong"] & R^{(M \times M)}
\end{tikzcd}
\]
to define $R^{(M)} \times R^{(M)} \to R^{(M)}$ as the unique bilinear map making the above diagram commute. Here $\otimes$ is the universal bilinear map:
\[ \otimes : (\sum_i r_i m_i, \sum_j s_j n_j) \mapsto \sum_{ij} r_i s_j (m_i,n_j) \]
giving rise to the standard formula
\[ (\sum_i r_i m_i) \cdot (\sum_j s_j n_j) \, = \, \sum_{ij} r_i s_j (m_i \cdot n_j) . \]
There is a subtlety in lifting this to the presheaf level, since there we have a \emph{dependent product}.

For the remainder of this section, we shall fix the following notation.
We are given $\Sh : A \to B$ and $\Tsh : B \to C$.
We define $\Sh \dprod \Tsh$ to be the presheaf on $\Sk(A)$ defined by $\sumST(U) := \{ (s,t) \mid s \in \Sh(U) \AND t \in \Tsh(\im s) \}$.
Restriction maps are defined as follows: given $U' \subseteq U$, $\rho^{U}_{U'} : (s,t) \mapsto (s |_{U'}, t |_{\im (s |_{U'})})$.
Functoriality can be verified straightforwardly.

We can now define a natural transformation $m : \sumST \natarrow \Tsh \circ \Sh$. For each $U \in \Sk(A)$, $m_U : (s,t) \mapsto t \circ s$. Given $U' \subseteq U$, naturality is the requirement that the following diagram commutes
\[ \begin{tikzcd}
\sumST (U) \ar[r, "m_U"] \ar[d, "\rho^{U}_{U'}"'] & \Tsh \circ \Sh(U) \ar[d, "\rho^{U}_{U'}"] \\
\sumST (U') \ar[r, "m_{U'}"]  & \Tsh \circ \Sh(U')
\end{tikzcd}
\]
which amounts to the equation $(t \circ s) |_{U'} \, = \, t |_{\im (s |_{U'})} \circ s |_{U'}$.

We note that the analogue of Proposition~\ref{compfamprop} holds for $\AbGrp$-valued presheaves:\footnote{In fact, the result holds for presheaves valued in any concrete category.}
\begin{proposition}
\label{gencompfamprop}
Let $\mathcal{F} : \Sk(A)^{\op} \to \AbGrp$.
There is a bijective correspondence between global sections  and $k$-compatible families of $\mathcal{F}$.
\end{proposition}
The proof is the same as that given for Proposition~\ref{compfamprop}.

\begin{proposition}\label{prodglobsecsprop}
Suppose we have global sections $\alpha : \One \natarrow \FZ \Sh$ with $\alpha_{U_0} = 1 \cdot s_0$, and $\beta : \One \natarrow \FZ \Tsh$ with $\beta_{\im s_0} = 1 \cdot t_0$. Then there is a global section $\alpha \dprod \beta : \One \natarrow \FZ (\sumST)$, with $(\alpha \dprod \beta)_{U_0} = 1 \cdot (s_0,t_0)$.
\end{proposition}
\begin{proof}
We define $(\alpha \dprod \beta)_{U} := \sum_{s \in \Sh(U)} \sum_{t \in \Tsh(\im s)} \alpha_s \beta_t (s,t)$, where 
$\alpha_U = \sum_{s \in \Sh(U)} \alpha_s s$ and $\beta_{\im s} = \sum_{t \in \Tsh(\im s)} \beta_t t$.

Clearly, $(\alpha \dprod \beta)_{U_0} = 1 \cdot (s_0,t_0)$. It remains to verify naturality.

Given $U' \subseteq U$, by naturality of $\alpha$ and $\beta$, for each $s' \in \Sh(U')$,  $t' \in \Tsh(\im s')$:
\[ \alpha_{s'} \; = \sum_{\substack{s \in \Sh(U), \\ s |_{U'} = s'}} \alpha_s, \qquad \beta_{t'} \; = \sum_{\substack{t \in \Tsh(\im s), \\ t |_{\im s'} = t'}} \beta_t . \]
Now given $s' \in \Sh(U')$ and $t' \in \Tsh(\im s')$, we compare the coefficient $\alpha_{s'}\beta_{t'}$ of $(s',t')$ in $(\alpha \dprod \beta)_{U'}$ with that in
$\FZ (\rho^{U}_{U'})((\alpha \dprod \beta)_{U})$. We have
\begin{align*}
\FZ (\rho^{U}_{U'})((\alpha \dprod \beta)_{U})(s',t') \; &= \sum_{\substack{s \in \Sh(U), \\ s |_{U'} = s'}} \; \sum_{\substack{t \in \Tsh(\im s), \\ t |_{\im s'} = t'}} \alpha_t \beta_t \\
&= \sum_{\substack{s \in \Sh(U), \\ s |_{U'} = s'}} \alpha_t \; (\sum_{\substack{t \in \Tsh(\im s), \\ t |_{\im s'} = t'}} \beta_t) \\
&= \sum_{\substack{s \in \Sh(U), \\ s |_{U'} = s'}} \alpha_t \beta_{t'} \\
&= (\sum_{\substack{s \in \Sh(U), \\ s |_{U'} = s'}} \alpha_t) \beta_{t'} \\
&= \alpha_{s'}\beta_{t'} .
\end{align*}
\end{proof}

\begin{proposition}
If $\Sh^{\Box} = \Sh$ and $\Tsh^{\Box} = \Tsh$, then $(\Tsh \circ \Sh)^{\Box} = \Tsh \circ \Sh$.
\end{proposition}
\begin{proof}
For $U_0 \in \Mk(A)$ and $t_0 \circ s_0 \in \Tsh \circ \Sh(U)$, we must show that $\Ztest(\Tsh \circ \Sh, t_0 \circ s_0)$. By Proposition~\ref{gencompfamprop}, it suffices to show that there is a global section $\gamma : \One \natarrow \FZ (\Tsh \circ \Sh)$ with $\gamma_U = 1 \cdot t_0 \circ s_0$. By assumption, and using Proposition~\ref{gencompfamprop} again, we have global sections $\alpha : \One \natarrow \FZ \Sh$ with $\alpha_{U_0} = 1 \cdot s_0$, and $\beta : \One \natarrow \FZ \Tsh$ with $\beta_{\im s_0} = 1 \cdot t_0$.
By Proposition~\ref{prodglobsecsprop}, we have $\alpha \dprod \beta : \One \natarrow \FZ (\sumST)$, with $(\alpha \dprod \beta)_{U_0} = 1 \cdot (s_0,t_0)$. Composing $\alpha \dprod \beta$ with $\FZ m : \FZ(\sumST) \natarrow \FZ(\Tsh \circ \Sh)$ yields $\gamma$ as required.
\end{proof}

As an immediate corollary, we obtain:
\begin{proposition}
The relation $\toZk$ is transitive.
\end{proposition}

\section{Cohomological reduction}\label{cohomredsec}

We now briefly outline the cohomological content of the construction $\Sh \mapsto \Shch$, which we shall call \emph{cohomological reduction}, since $\Shch \inc \Sh$.
We refer to \cite{DBLP:conf/csl/AbramskyBKLM15} for further details.

Given a flasque subpresheaf $\Sh$ of $\Hk$, we have the abelian-group-valued presheaf $\FF = \FZ \Sh$. We use the \Cech cohomology with respect to the cover $\MM = \Mk(A)$.
In order to focus attention at the context $C \in \MM$, we use the presheaf $\FF |_{C}$, which ``projects'' onto $C$. The cohomology of this presheaf is the \emph{relative cohomology} of $\FF$ at $C$.
The $i$'th relative \Cech cohomology group of $\FF$ is written as $\Cohom{i}$.
We have the \emph{connecting homomorphism} $\Cohom{0} \to \Cohom{1}$ constructed using the Snake Lemma of homological algebra.
The cohomological obstruction $\gamma : \FF(C) \to \Cohom{1}$ defined in \cite{DBLP:conf/csl/AbramskyBKLM15} is this connecting homomorphism, composed with the isomorphism $\FF(C) \cong \Cohom{0}$.

Using Proposition~\ref{gammacompfamprop}, the predicate $\Ztest(\Sh, s)$ is equivalent to $\gamma \circ \eta(s) = 0$, \ie $\eta(s) \in \ker \gamma$.
Thus we can define $\Shch(C)$ as the pullback (in $\Set$):
\[ \begin{tikzcd}
\Shch(C) \ar[d, hookrightarrow]  \arrow[dr, phantom, "\lrcorner", very near start] \ar[r] & U (\ker \gamma) \ar[d, hookrightarrow] \\
\Sh(C) \ar[r, "\eta"] & U \FZ \Sh (C)
\end{tikzcd}
\]

\section{Relation to contextuality conditions}

We recall  one of the contextuality properties studied in \cite{DBLP:conf/csl/AbramskyBKLM15}.
In the present setting, we can define this as follows. If $\Sh$ is a flasque subpresheaf of $\Hk$, then:
\begin{itemize}
\item $\Sh$ is \emph{cohomologically strongly contextual}  $(\CSC(\Sh))$ if
 \[ \forall C \in \Mk(A). \, \forall s \in \Sh(C). \, \neg \Ztest(s) . \]
\end{itemize}

We shall write $A \ntoZk B \, := \, \neg(A \toZk B)$, and similarly $A \ntoZko B$.
\begin{proposition}\label{CSCntoZkoprop}
$\CSC(\Sbar) \IMP A \ntoZko B$.
\end{proposition}
\begin{proof}
If $\CSC(\Sbar)$, then $\Sbar^{\Box}(C) = \es$ for all $C \in \Mk(A)$. This implies that $\Sbar^{\Box\Diamond}(C) = \es$ for all $C \in \Sk(A)$, \ie $\Shone = \mathbf{\es}$, and hence $A \ntoZko B$.
\end{proof}
\section{Linear templates: the power of one iteration}
We now consider the case where the template structure $B$ is linear. This means that $B = R$ is a finite ring, and the interpretation of each relation in $\sg$ on $R$ has the form $\Eab^{R}(r_1,\ldots,r_n) \equiv \sum_{i=1}^n a_i r_i = b$, for some $\va \in R^n$ and $b \in R$. Thus we can label each relation in $\sg$ as $\Eab$, where $\va$, $b$ correspond to the interpretation of the relation in $R$.

Given an instance $A$, we can regard each tuple $\vx \in A^n$ such that $\Eab^{A}(x_1,\ldots, x_n)$ as the equation $\sum_{i=1}^n a_i x_i = b$. The set of all such equations is denoted by $\TA$.
We say that a function $f : A \to R$ \emph{satisfies} this equation if $\sum_{i=1}^n a_i f(x_i) = b$ holds in $R$, \ie if  $\Eab^{A}(f(x_1),\ldots, f(x_n))$. It is then immediate that a function $f : A \to R$ simultaneously satisfies all the equations in $\TA$ iff it is a homomorphism.

We can also associate a set of equations with each context $C \in \Mk(A)$. We say that $\Sbar(C)$ satisfies the equation $e_{\vc,d} := \sum_{i=1}^n c_i x_i = d$ if for all $s \in \Sbar(C)$, $s$ satisfies $e_{\vc,d}$, where $\{ x_1, \dots , x_n \} \subseteq C$,
$\vc \in R^n$ and $d \in R$. 
Note that we do \emph{not} require that there is a corresponding relation $E_{\vc,d}$ in $\sg$. We write $\TC$ for the set of all equations satisfied by $\Sbar(C)$, and $\TS := \bigcup_{C \in \Mk} \TC$.
We say that $\Sbar$ satisfies \emph{All-versus-Nothing contextuality} ($\AvN_{R}(\Sbar)$) \cite{DBLP:conf/csl/AbramskyBKLM15} if there is no function $s : A \to R$ which satisfies all the equations in $\TS$.
\begin{proposition}
\label{ntoavnprop}
If $A \nto R$, then $\AvN_{R}(\Sbar)$.
\end{proposition}
\begin{proof}
Since we are assuming that $k \geq n$, where $n$ is the relational width of $\sg$, we have $\TA \subseteq \TS$, and hence any function satisfying $\TS$ will also satisfy $\TA$. As already observed, $A \nto B$ is equivalent to the statement that there is no satisfying assignment for $\TA$.
\end{proof}

We can now state another  important result from \cite{AOC2021}: that cohomological $k$-consistency is an \emph{exact condition} for linear templates.
Moreover, the key step in the argument is the main result from  \cite{DBLP:conf/csl/AbramskyBKLM15}.
\begin{proposition}
\label{linearprop}
For every linear template $R$, and instance $A$:
\[ A \to R \; \IFF \; A \toZk R \; \IFF \; A \toZko R . \]
\end{proposition}
\begin{proof}
The forward implications are given by Proposition~\ref{chainimpprop}.
Now suppose that $A \nto R$. By Proposition~\ref{ntoavnprop}, this implies that $\AvN_{R}(\Sbar)$. By \cite[Theorem 6.1]{DBLP:conf/csl/AbramskyBKLM15}, this implies that $\CSC(\Sbar)$. 
By Proposition~\ref{CSCntoZkoprop}, this implies $A \ntoZko B$. By Proposition~\ref{chainimpprop} again, this implies $A \ntoZk R$.
\end{proof}

\section{Cohomological width}

We recall that a CSP template structure $B$ is said to have \emph{width} $\leq k$ (written as $\Wk(B)$) if, for all instances $A$:
\[ A \to B \; \IFF \; A \tok B . \]
Thus the templates of bounded width are those for which $\CSP(B)$ has an exact polynomial-time solution given by determination of strong $k$-consistency for some $k$.

In their seminal paper \cite{feder1998computational} Feder and Vardi identified two tractable subclasses of CSP, those with templates of bounded width, and those which are ``subgroup problems'' in their terminology, \ie essentially those with linear templates. Since all other cases with known complexity at that time were NP-complete, this motivated their famous Dichotomy Conjecture, recently proved by Bulatov \cite{bulatov2017dichotomy} and Zhuk \cite{zhuk2020proof}.

The two tractable classes identified by Feder and Vardi appeared to be quite different in character. However, we can use the preceding cohomological analysis to give a unified account.
We define the \emph{cohomological width} of a template structure $B$ to be $\leq k$ (notation: $\WZk(B)$) if, for all instances $A$:
\[ A \to B \; \IFF \; A \toZk B . \]
\begin{proposition}
If $B$ has bounded cohomological width, \ie $\WZk(B)$ for some $k$, then $\CSP(B)$ is in PTIME.
Moreover, both the Feder-Vardi classes of bounded width and linear templates have bounded cohomological width.
\end{proposition}
\begin{proof}
The first statement follows from Proposition~\ref{ckconprop}. Note that this provides a simple, uniform algorithm which gives an exact solution for $\CSP(B)$ whenever $B$ has bounded cohomological width.
For the second statement, from Proposition~\ref{chainimpprop} we have the implication $\Wk(B) \Rightarrow \WZk(B)$, hence bounded width implies bounded cohomological width.
Finally, Proposition~\ref{linearprop} implies that linear templates have bounded cohomological width.
\end{proof}

For each template structure $B$, either $\CSP(B)$ is NP-complete, or $B$ admits a weak near-unanimity polymorphism  \cite{maroti2008existence}. 
In \cite{zhuk2020proof}, Zhuk shows that if $B$ admits a weak near-unanimity polymorphism, there is a polynomial-time algorithm for $\CSP(B)$, thus establishing the Dichotomy Theorem.
This result motivates the following question (see also \cite{AOC2021}):
\begin{question}
Is is the case that for all structures $B$, if $B$ has a weak near unanimity polymorphism, then it has bounded cohomological width?
\end{question}
A positive answer to this question would give an alternative proof of the Dichotomy Theorem.

\section{Presheaf representations of logical equivalences}

As we have seen, the local consistency relation $A \tok B$ approximates the homomorphism relation $A \to B$.
By standard results (cf.~Proposition~\ref{pebbprop} and \cite{abramsky2017pebbling}), $A \tok B$ iff every $k$-variable existential positive sentence satisfied by $A$ is satisfied by $B$. We now consider presheaf representations of logical equivalence for richer logical fragments.

\subsection{Existential logic $\ELk$}

As a first step, we consider $\Ik$, the presheaf of \emph{partial isomorphisms}. For each $C \in \Sk(A)$, $\Ik(C)$ is the set of partial isomorphisms from $A$ to $B$ with domain $C$. This is a subpresheaf of $\Hk$.

We can now consider $\Ikfl$, the coflasquification of $\Ik$.
We have the following analogue of Proposition~\ref{homglobsecsprop}:
\begin{proposition}
Suppose that $k \geq n$, where $n$ is  the maximum arity of any relation in $\sg$. There is a bijective correspondence between
\begin{enumerate}
\item embeddings $A \embed B$
\item global sections of $\Ikfl$.
\end{enumerate}
\end{proposition}

Note that $\Ikfl$ can be computed in polynomial time, by an algorithm which is a minor variation of the local consistency algorithm which produces $\Hk^{\fl}$.
We can see a non-empty flasque subpresheaf of $\Ik$ as a winning Duplicator strategy for the forth-only $k$-pebble game in which the winning condition is the partial isomorphism condition rather than the partial homomorphism condition.

This leads to the following correspondence with the logical preorder induced by \emph{existential logic} $\ELk$. This is the $k$-variable logic which allows negation on atomic formulas, as well as conjunction, disjunction, and existential quantification.
For structures $A$, $B$, $A \ELkpreord B$ if every sentence of $\ELk$ satisfied by $A$ is also satisfied by $B$.
\begin{proposition}
For finite structures $A$, $B$, the following are equivalent:
\begin{enumerate}
\item $\Ikfl \neq \es$
\item $A \ELkpreord B$.
\end{enumerate}
\end{proposition}

\subsection{The full $k$-variable fragment $\Lk$}
To express back-and-forth equivalences, we use the fact that the inverse operation on partial isomorphisms lifts to $\Ik$ and its subpresheaves. 
Given a subpresheaf $\Sh$ of $\Ik = \IkAB$, we define $\Shdag$ to be the  subpresheaf of $\IkBA$ on $\Sk(B)$ such that, for $D \in \Sk(B)$, $\Shdag(D) = \{ t \in \IkBA(D) \mid t^{-1} \in \Sh(\im t) \}$.
The fact that restriction is well-defined in $\Shdag$ follows since, if $s \in \Ik(C)$, $D = \im s$, and $D' \subseteq D$,
then $s^{-1} |_{D'} = (s |_{C'})^{-1}$, where $C' = \im \, (s^{-1} |_{D'})$.

\begin{proposition}
For finite structures $A$, $B$, the following are equivalent:
\begin{enumerate}
\item There is a non-empty subpresheaf $\Sh$ of $\Ik$ such that both $\Sh$ and $\Shdag$ are flasque.
\item $A \eqLk B$.
\end{enumerate}
\end{proposition}

\subsection{The $k$-variable fragment with counting quantifiers $\Lck$}\label{cofpsec}

We recall Hella's bijection game \cite{Hella1996}: at each round, Duplicator specifies a bijection $\beta$ from $A$ to $B$, and then Spoiler specifies 
$a \in A$. The current position is extended with $(a, \beta(a))$. The winning condition is the partial isomorphism condition.
Note that allowing Spoiler to choose $b \in B$ would make no difference, since they could have chosen $a = \beta^{-1}(b)$ with the same effect.

Given a flasque subpresheaf $\Sh$ of $\Ik$, the additional condition corresponding to the bijection game is that, for each $s \in \Sh(C)$ with $|C| < k$, there is an assignment $a \mapsto s_a$ of $s_a \in \Sh(C \cup \{a\})$ with $s_a |_C = s$, such that the relation $\{ (a,s_a(a)) \mid a \in A \}$ is a bijection from $A$ to $B$.  We write $\cotest(\Sh,s)$ if this condition is satisfied for $s$ in $\Sh$.
As observed in \cite{AOC2021}, for each such $s$ this condition can be formulated as a perfect matching problem on a bipartite graph of size polynomial in $|A|$ and $|B|$, and hence can be checked in polynomial time.

We lift this test to the presheaf $\Sh$, and define $\Shcotest \inc \Sh$ by $\Shcotest(C) = \{ s \in \Sh(C) \mid \cotest(\Sh,s) \}$.
Since there are only polynomially many $s$ to be checked, $\Shcotest$ can be computed in polynomial time.

We can then define an iterative process
\[ \Ik \linc \Ik^{\fl} \linc \Ik^{\fl\cotest\fl} \linc \cdots \linc \Ik^{\fl(\cotest\fl)^{m}} \linc \cdots \]
This converges to a fixpoint $\Shcofp$ in polynomially many iterations. 
\begin{proposition}
For finite structures $A$, $B$, the following are equivalent:
\begin{enumerate}
\item $\Shcofp \neq \es$.
\item $A \eqLck B$.
\end{enumerate}
\end{proposition}

By a well-known result \cite{cai1992optimal}, the logical equivalence $\eqLck$ corresponds to equivalence at level $k-1$ according to the Weisfeiler-Leman algorithm, a widely studied approximation algorithm for graph and structure isomorphism \cite{kiefer2020weisfeiler}.
Thus the above yields an alternative polynomial-time algorithm for computing Weisfeiler-Leman equivalences.

\section{Cohomological refinement of logical equivalences}

We can proceed analogously to the introduction of a cohomological refinement of local consistency in Section~\ref{cohomkconsec}. Such refinements are possible for all the logical equivalences studied in the previous Section.
We will focus on the equivalences $\eqLck$, and the corresponding Weisfeiler-Leman equivalences.

The cohomological reduction  $\Sh \mapsto \Shch$ is directly applicable to flasque subpresheaves of $\Ik$. 
However, an additional subtlety which arises in this case is that $\Ztest(\Sh,s)$ need not imply $\Ztest(\Shdag, s^{-1})$ in general. We are thus led to the following symmetrization of cohomological reduction:
$\Sh \mapsto \Shsch := \Sh^{\ch \dagger \ch \dagger}$.

Combining this with the counting reduction fixpoint  $\Sh \mapsto \Sh^{\cofp}$ from Section~\ref{cofpsec} leads to an iterative process
\[ \Ik \linc \Ik^{\cofp} \linc \Ik^{\cofp\sch\cofp} \linc \cdots \linc \Ik^{\cofp(\sch\cofp)^{m}} \linc \cdots \]
This converges to a fixpoint $\Shcst$ in polynomially many iterations. Since each of the operations $\Sh \mapsto \Shsch$ and $\Sh \mapsto \Sh^{\cofp}$ is computable in polynomial time, so is  $\Shcst$.

We define $A \eqZ B$ if $\Shcst \neq \es$. We can view this as a polynomial time approximation to structure isomorphism.
The soundness of this algorithm is expressed as follows.
\begin{proposition}
For all finite structures $A$, $B$:
\[ A \cong B \IMP A \eqZ B \IMP A \eqLck B . \]
\end{proposition}
In \cite{AOC2021}, it is shown that $\eqZ$ strictly refines $\eqLck$. Moreover, Proposition~\ref{linearprop} is leveraged to show that $\eqZ$ is discriminating enough to defeat two important families of counter-examples: the CFI construction used in \cite{cai1992optimal} to show that $\Lck$ is not strong enough to characterise polynomial time, and the constructions in \cite{lichter2021separating,dawar2021limitations} which are used to show similar results for linear algebraic extensions of $\Lck$.

\section{A meta-algorithm for presheaves}

It is clear that the various algorithms we have described for local consistency, logical equivalences, and their cohomological refinements have a common pattern.
This pattern has two  ingredients:
\begin{enumerate}
\item An \emph{initial over-approximation} $\Sh_0$.
\item a \emph{local predicate} $\vphi$.
\end{enumerate}
We can use $\vphi$ to define a \emph{deflationary operator} $J = \Jphi$ on the sub-presheaves of $\Sh_0$.

As a general setting, we take presheaves $\Pshv{P}$ on a poset $P$, ordered by sub-presheaf inclusion. For any choice of initial presheaf $\Sh_0$, the downset ${\downarrow} \Sh_0$ forms a complete lattice $L$.\footnote{Equivalently, this is the subobject lattice $\Sub(\Sh_0)$.}
A deflationary operator on $L$ is a monotone function $J : L \to L$ such that, for all $\Sh \in L$, $J \Sh \inc \Sh$. By the standard Tarski fixpoint theorem, this has a greatest fixpoint $J^* = \bigcup \{ \Sh \mid \Sh \inc J \Sh \}$.
To compute this greatest fixpoint, if $L$ satisfies the Descending Chain Condition, \ie there are no infinite descending chains, then we can compute
\[ \Sh_0 \linc J \Sh_0 \linc J^2 \Sh_0 \linc \cdots \]
which will converge after finitely many steps to $J^*$.

Now we specialize to subpresheaves of $\HkAB$. We define $|\Sh| := \sum_{C \in \Sk(A)} |\Sh(C)|$. Since $|\HkAB| \leq |A|^k|B|^k$, the number of iterations needed to converge to $J^*$ is polynomially bounded in $|A|$, $|B|$.
We consider deflationary operators of the form $J = J_{\vphi}$, where $\vphi(\Sh,s)$ is a predicate on local sections.
We require that $\vphi$ be monotone in the following sense: if $\Sh \inc \Sh'$, then for $C \in \Sk(A)$, $s \in \Sh(C)$, $\vphi(\Sh,s) \Rightarrow \vphi(\Sh',s)$.
We define $J_{\vphi}( \Sh)(C) := \{ s \in \Sh(C) \mid \vphi(\Sh,s) \}$. Since $|\Sh|$ is polynomially bounded in $|A|$, $|B|$, $J_{\vphi}$ makes polynomially many calls of $\vphi$.

We recover all the algorithms considered previously by making suitable choices for the predicate $\vphi$.
In most cases, this can be defined by $\vphi(\Sh, s) \; \equiv \; \bar{\Sh}, s \models \psi$, where $\bar{\Sh}$ is the relational structure with universe $\sum_{C \in \Sk(A)} \Sh(C)$, and the relation 
\[ E^{\bar{\Sh}}(s,t,a,b) \equiv (t = s \cup \{ (a,b) \});  \]
while $\psi$ is a first-order formula over this vocabulary.\footnote{Strictly speaking, we should consider the many-sorted structure with sorts for $A$ and $B$.}
The counting predicate $\cotest$ and the cohomological reduction predicate $\Ztest$ require stronger logical resources.
In all cases, if $\vphi$ is computable in polynomial time, so is the overall algorithm to compute $J^*$.

\section{Further remarks and questions}

Firstly, we observe a down-closure property of cohomological $k$-consistency.
\begin{proposition}
If $k \leq k'$, then $A \toZkp B \IMP A \toZk B$.
\end{proposition}
\begin{proof}
Suppose we have a $\ZZ$-compatible family $\{ \alpha_{C'} \}_{C' \in \Mkp(A)}$ extending $s' \in \Sbarp(C'_0)$. By  Proposition~\ref{gencompfamprop} we can extend this family to a global section of $\FZ \Sbarp$.
Using Proposition~\ref{gencompfamprop} again, the elements $\{ \alpha_{C} \}_{C \in \Mk(A)}$ at the $k$-element subsets will form a $\ZZ$-compatible family of $\Sbar$. Moreover, if $C_0$ is a $k$-element subset of $C'_0$, $\alpha_{C_0}$ will be $s = s' |_{C_0}$, so this family will extend $s$.
\end{proof}

We now ask which of the implications in Proposition~\ref{chainimpprop} can be reversed. Clearly, if $P \neq NP$, the first cannot be reversed, but there should be an explicit counter-example.
There are such counter-examples of false negatives for contextuality in \cite{DBLP:journals/corr/abs-1111-3620,DBLP:conf/csl/AbramskyBKLM15,caru2017cohomology}, but these are for specific values of $k$. Can we provide a family of counter-examples for all values of $k$?
We try to phrase this more precisely:
\begin{question}
If we fix an NP-complete template $B$, e.g. for 3-SAT, can we find a family of instances $\{ A_k \}$ such that, for all $k$:
\[ A_k \toZk B \AND A_k \nto B . \]
\end{question}

\bibliographystyle{amsplain}
\bibliography{bibfile}

\end{document}